\documentclass[aps,pra,superscriptaddress,amsmath,amssymb,twocolumn,footinbib]{revtex4-1}

\usepackage[T1]{fontenc}
\usepackage[utf8]{inputenc}

\usepackage{graphicx}
\usepackage[dvipsnames,svgnames]{xcolor}
\usepackage{braket}
\usepackage{newtxtext,newtxmath}

\usepackage{bbm,amsthm,amsmath}
\usepackage{enumerate}
\usepackage{subcaption,caption}
\usepackage{booktabs,multirow}
\usepackage{mathtools,mathrsfs}
\usepackage[percent]{overpic}
\usepackage{microtype,hyphenat}
\usepackage{array,adjustbox,afterpage,breakurl}

\newtheoremstyle{red}{}{}{\itshape}{}{\color{red!80!black}\bfseries}{.}{ }{}

\definecolor{darkred}{rgb}{0.57,0,0.12}
\usepackage[colorlinks=true, linkcolor=MidnightBlue, urlcolor=darkred!75!black, citecolor=MidnightBlue, breaklinks=true]{hyperref}

\newcommand{\nc}{\newcommand}
\usepackage{cleveref}
\usepackage[most,breakable]{tcolorbox}

\AtBeginDocument{
\heavyrulewidth=.08em
\lightrulewidth=.05em
\cmidrulewidth=.03em
\belowrulesep=.65ex
\belowbottomsep=0pt
\aboverulesep=.4ex
\abovetopsep=0pt
\cmidrulesep=\doublerulesep
\cmidrulekern=.5em
\defaultaddspace=.5em
}

\usepackage{etoolbox}
\hypersetup{linktoc=all}

\nc{\ketbra}[2]{\ket{#1}\bra{#2}}

\DeclareMathOperator{\Tr}{Tr}

\DeclareMathOperator{\NN}{\mathbb{N}}

\DeclareMathOperator{\rank}{rank}

\newcommand{\norm}[2]{\left\lVert#1\right\rVert_{\,#2}}
\newcommand{\proj}[1]{\ket{#1}\!\bra{#1}}

\nc{\TT}[1]{T^{({#1})}_\I}
\nc{\TTJ}[1]{T^{({#1})}_\J}

\newcommand{\T}{\mathcal{T}}

\newcommand{\I}{\mathcal{I}}

\newcommand{\<}{\left\langle}
\renewcommand{\>}{\right\rangle}

\renewcommand{\bar}{\;\rule{0pt}{9.5pt}\right|\;}

\newcommand{\lset}{\left\{\left.}
\newcommand{\rset}{\right\}}

\newcommand{\DD}{\mathbb{D}}

\newcommand{\ve}{\varepsilon}

\newcommand{\cbraket}[1]{\left|\braket{#1}\right|}

\newcommand{\id}{\mathbbm{1}}

\newcommand{\J}{\mathcal{J}}

\nc{\PPTPPR}{\text{\rm PPTP}_+}
\nc{\PPTPR}{\text{\rm PPTP}}

\nc{\ppt}{\text{\rm\sffamily PPT}}
\nc{\pptp}{\text{\rm\sffamily PPT}_{+}}

\newcommand{\PPT}{\ppt}

\nc{\PPTRP}{{{\PPT^{\hspace{0.1em}\prime}_+}}}
\nc{\PPTRPPR}{{\text{\rm PPTP}^{\hspace{0.1em}\prime}_+}}

\nc{\wt}{\widetilde}

\newcommand*{\cE}{\mathcal{E}}
\newcommand*{\cH}{\mathcal{H}}

\newcommand*{\cL}{\mathcal{L}}

\newcommand*{\cV}{\mathcal{V}}
\newcommand{\vect}[1]{\mathbf{#1}}
\def\gbm#1{{\let\phi\upphi \let\lambda\uplambda \let\mu\upmu \let\rho\uprho \let\sigma\upsigma \let\tau\uptau \let\theta\uptheta \let\eta\upeta \let\kappa\upkappa \bm{#1}}}
\newcommand{\vecg}[1]{\gbm{#1}}
\newcommand\precw{\mathrel{\stackrel{\makebox[0pt]{\mbox{\normalfont\tiny weak}}}{\prec}}}
\newcommand{\be}{\begin{equation}}
\newcommand{\ee}{\end{equation}}
\newcommand{\n}{\textendash}

\newcommand*{\bbC}{\mathbb{C}}

\newcommand*{\fr}[2]{\frac{#1}{#2}}

\nc{\logfloor}[1]{\left\lfloor {#1} \right\rfloor_{\log}}
\nc{\logceil}[1]{\left\lceil {#1} \right\rceil_{\log}}

\newtheorem{theorem}{Theorem}
\newtheorem{proposition}[theorem]{Proposition}
\newtheorem{corollary}[theorem]{Corollary}
\newtheorem{definition}[theorem]{Definition}

\newtheorem{thm}{Theorem}

\newtheorem{obsm}[thm]{Observation}

\theoremstyle{red}

\theoremstyle{definition}
\newtheorem*{remark}{Remark}

\let\oldproofname\proofname
\renewcommand{\proofname}{\rm\bf{\oldproofname}}

\let\nc\newcommand
  \nc{\MIO}{{\text{\rm MIO}}}
\nc{\DIO}{{\text{\rm DIO}}}
\nc{\SIO}{{\text{\rm SIO}}}
\nc{\IO}{{\text{\rm IO}}}
\nc{\lsetr}{\left\{\,}
\nc{\rsetr}{\right.\right\}}
\nc{\barr}{\,\rule{0pt}{9.5pt}\left|\;}

\nc{\rhodio}{{\rho\text{\rm-DIO}}}
\nc{\psidio}{{\psi\text{\rm-DIO}}}
\nc{\psimdio}{{\Psi_m\text{\rm-DIO}}}
\newcommand{\mlt}{\prec}

\begin{document}

\title{Coherence manipulation with dephasing-covariant operations}

\author{Bartosz Regula}
\email{bartosz.regula@gmail.com}
\affiliation{School of Physical and Mathematical Sciences, Nanyang Technological University, 637371, Singapore}
\affiliation{Complexity Institute, Nanyang Technological University, 637335, Singapore}

\author{Varun Narasimhachar}
\email{nvarun@ntu.edu.sg}
\affiliation{School of Physical and Mathematical Sciences, Nanyang Technological University, 637371, Singapore}
\affiliation{Complexity Institute, Nanyang Technological University, 637335, Singapore}

\author{Francesco Buscemi}
\email{buscemi@i.nagoya-u.ac.jp}
\affiliation{Graduate School of Informatics, Nagoya University, Chikusa-ku, 464-8601 Nagoya, Japan}

\author{Mile Gu}
\email{mgu@quantumcomplexity.org}
\affiliation{School of Physical and Mathematical Sciences, Nanyang Technological University, 637371, Singapore}
\affiliation{Complexity Institute, Nanyang Technological University, 637335, Singapore}
\affiliation{Centre for Quantum Technologies, National University of Singapore, 3 Science Drive 2, 117543, Singapore}

\begin{abstract}%
We characterize the operational capabilities of quantum channels which can neither create nor detect quantum coherence vis-\`{a}-vis efficiently manipulating coherence as a resource. We study the class of dephasing-covariant operations (DIO), unable to detect the coherence of any input state, as well as introduce an operationally-motivated class of channels $\rho$-DIO which is tailored to a specific input state.
We first show that pure-state transformations under DIO are completely governed by majorization, establishing necessary and sufficient conditions for such transformations and adding to the list of operational paradigms where majorization plays a central role. We then show that $\rho$-DIO are strictly more powerful: although they cannot detect the coherence of the input state $\rho$, the operations $\rho$-DIO can distill more coherence than DIO. However, the advantage disappears in the task of coherence dilution as well as generally in the asymptotic limit, where both sets of operations achieve the same rates in all transformations.
\end{abstract}

\maketitle

\section{Introduction}
Quantum coherence, or superposition, is an intrinsic feature of quantum mechanics that underlies the advantages enabled by quantum information processing and quantum technologies~\cite{streltsov_2017}. The resource theory of quantum coherence~\cite{aberg_2006,baumgratz_2014,winter_2016,streltsov_2017} has found extensive use in the characterization of our ability to manipulate coherence efficiently within a rigorous theoretical framework, wherein the properties of a resource are investigated under a suitable set of allowed (``free'') operations which reflect the constraints placed on the manipulation of the given resource~\cite{horodecki_2012,chitambar_2019}. Despite many promising developments in the establishment of a comprehensive description of coherence, the physical constraints governing its manipulation are unsettled~\cite{chitambar_2016,marvian_2016,streltsov_2017}, and one of the most important outstanding questions in the resource theory of quantum coherence is to understand the exact properties of the different operational paradigms under which coherence can be investigated~\cite{winter_2016,du_2015, yuan_2015, yadin_2016, chitambar_2016, marvian_2016, chitambar_2016-1, vicente_2017, liu_2017,zhu_2017,streltsov_2017-2, zhao_2018, regula_2017, chitambar_2018, egloff_2018, fang_2018, theurer_2019, regula_2018-1, zhao_2019, lami_2019-1,lami_2019}.

Many proposed types of free operations stem from meaningful physical considerations regarding their implementability: for instance, the physically incoherent operations~\cite{chitambar_2016} only require the use of incoherent ancillary systems and incoherent measurements, making them cheap and easy to implement in the resource-theoretic setting; the strictly incoherent operations (SIO)~\cite{winter_2016,yadin_2016} allow for a similar implementation with incoherent ancillae, but require arbitrary measurements. Regrettably, these operations were found to be too limited in their operational capabilities \cite{vicente_2017,lami_2019-1,lami_2019}, suggesting that any non-trivial resource theory of coherence would require a larger set of allowed maps in order to be practically useful. Common choices of such larger sets, however, are often too permissive and lack a physical justification~\cite{chitambar_2016}.

The question of ``freeness'' of operations within the resource theory of coherence has also been approached from a different perspective, by characterizing the nonclassicality of the channels themselves --- first in a very general manner~\cite{meznaric_2013}, and later with explicit applications to coherence theory~\cite{yadin_2016,liu_2017,theurer_2019}. These works suggest that the most important feature distinguishing truly nonclassical channels is their ability to \textit{detect} the coherence present in a system and use it to manipulate the state, and so any class of channels which is to be considered free must necessarily be incapable of coherence detection. Although the study of maps based on various types of coherence non-detecting properties has attracted significant attention in the recent years \cite{chitambar_2016,marvian_2016,yadin_2016,liu_2017,zhao_2018,regula_2017,fang_2018,chitambar_2018,theurer_2019}, the limits of their power remain undiscovered.

In this work, we explore the limitations imposed on coherence manipulation by the inability of the free operations to create and detect coherence. 
We first study the class of \emph{dephasing-covariant incoherent operations} (DIO)~\cite{chitambar_2016,marvian_2016}, constituting the largest class of operations that do not detect the coherence of any input state. We establish a complete description of pure-state transformations under these operations by relating them with the theory of majorization, revealing also an operational connection between DIO and other classes of free operations. To investigate the ultimate operational limits of coherence non-detecting channels, we then introduce the class of operations $\rho$-DIO, which are tailored to a specific input state and extend the class DIO. In particular, we quantify the capabilities of $\rho$-DIO in the fundamental operational tasks of coherence distillation and dilution. We show in particular that such maps satisfy a curious property: even though $\rho$-DIO cannot detect the coherence of the state $\rho$, they can still \textit{distill} more coherence from $\rho$ than the class DIO; however, this advantage disappears in the one-shot task of dilution as well as generally in the asymptotic limit, where DIO match the capabilities of $\rho$-DIO in all state transformations. Our results give insight into the precise limits on the operational power of free operations which do not use coherence in manipulating coherence as a resource.

\section{DIO and $\rho$-DIO}
We study quantum coherence as a basis-dependent concept. We will therefore fix an orthonormal basis $\{\ket{i}\}_{i=1}^d$ which we deem incoherent, and use $\I$ to denote the set of all states diagonal (incoherent) in this basis. We will use $\Delta(\cdot) = \sum_i \proj{i} \cdot \proj{i}$ to denote the completely dephasing channel in this basis.

As discussed previously, it is natural to impose two conditions for a set of operations to be free: one is its inability to \textit{create} coherence, in the sense that $\sigma \in \I \Rightarrow \Lambda(\sigma) \in \I$, and the other is its inability to \textit{detect} coherence, in the sense that measurement statistics under any incoherent measurement after an operation $\Lambda$ should remain the same regardless of whether the input state possessed any coherence or not; precisely, $\braket{i | \Lambda(\rho) | i} = \braket{i | \Lambda(\Delta(\rho)) | i}$ for all $i$. Imposing only the first constraint leads to the class of maximally incoherent operations (MIO)~\cite{aberg_2006}, which can exhibit undesirable properties such as being able to increase the number of levels of a pure state which are in superposition~\cite{chitambar_2016-1}. The class of channels which satisfy both constraints for all input states is precisely DIO, equivalently defined to commute with the dephasing channel: $\Lambda \circ \Delta(\rho) = \Delta \circ \Lambda (\rho) \; \forall \rho$.
These operations have previously been considered in various contexts \cite{meznaric_2013,chitambar_2016,marvian_2016,liu_2017}, and indeed they admit several interpretations. The maps DIO can be regarded as inherently classical \cite{meznaric_2013,theurer_2019}, as any classical (incoherent) observer is unable to distinguish $\Lambda(\rho)$ from $\Lambda\circ\Delta(\rho)$, and hence is unable to determine whether the coherence of $\rho$ has been employed in the process. The latter point shows that DIO can also be understood as the operations which do not use coherence \cite{yadin_2016,liu_2017}, as the properties of the output system accessible to a classical observer are independent of the coherence of the input.

The ability to detect coherence is of particular importance in practical setups relying on quantum coherence, such as interferometric experiments \cite{yadin_2016,biswas_2017,theurer_2019}. A general interferometric protocol can be understood as consisting of three separate parts: first, a state in superposition is created; second, path-dependent phases are encoded in the state with suitable unitary operations; and third, the information about the paths is extracted in a measurement. It is then explicit that the ability to create (in the first step) and detect (in the last step) coherence are crucial for any such setup to work, and indeed any operation which can neither create nor detect coherence is inherently free and cannot be used in such an experimental protocol.

However, consider now a scenario in which the input coherent state $\rho$ of a protocol is known: the operations which cannot detect the coherence of the input state are then precisely those which satisfy $\Lambda \circ \Delta (\rho) = \Delta \circ \Lambda(\rho)$ for this choice of $\rho$, and indeed it is not necessary to impose dephasing covariance for all inputs if one is concerned with acting on $\rho$ specifically. This point of view motivates us to define the class of $\mathbf{\rho}$\emph{--dephasing-covariant incoherent operations} ($\rho$-DIO), which we take to be the operations which commute with the dephasing channel $\Delta$ for a given input state $\rho$, and which therefore incorporate the ultimate limitations caused by the inability to detect or use the coherence of a particular input state.

It is clear that a $\rho$-DIO map can in principle create or detect coherence when acting on an input state other than $\rho$. However, the definition of $\rho$-DIO is justified whenever one deals with an explicit protocol which transforms a fixed input state to some desired output. Two such protocols form the foundations of the manipulation of coherence as a quantum resource: these are the tasks of \textit{coherence distillation} \cite{winter_2016,regula_2017,zhao_2019}, which aims to convert a given input state to a maximally coherent state, and \textit{coherence dilution} \cite{winter_2016,zhao_2018}, which performs the opposite transformation of a maximally coherent input state to some desired state. The definition of $\rho$-DIO then motivates the question: can the operational capabilities of DIO be surpassed by operations which nevertheless do not detect the coherence of the input state $\rho$? To address this question, we first describe the transformations achievable under DIO, and later investigate whether $\rho$-DIO can outperform the former.

\section{Pure-state transformations under DIO}
Although a fundamental and operationally meaningful choice of operations, the class DIO remains relatively unexplored, and few of its properties are known. Other sets of operations are better understood: in particular, it is known that the transformations of pure states under the classes of incoherent operations (IO) \cite{baumgratz_2014} and strictly incoherent operations (SIO) \cite{winter_2016,yadin_2016} are governed by majorization theory, in a manner similar to the manipulation of pure-state entanglement under local operations and classical communication \cite{nielsen_1999}. Precisely, one has that a pure-state transformation $\ket\psi = \sum_i \psi_i \ket{i} \to \ket\phi = \sum_i \phi_i \ket{i}$ is achievable under IO or SIO if and only if $\Delta(\psi) \mlt \Delta(\phi)$ \cite{chitambar_2016-1,du_2015-1,zhu_2017-1}, i.e. if $\sum_{i=1}^k |\psi_i|^2 \leq \sum_{i=1}^k |\phi_i|^2 \; \forall k \in \{1, \ldots, d\}$ where we assume that the coefficients of the states are arranged so that $|\psi_1| \geq \ldots \geq |\psi_d|$. Our first contribution is to extend this relation to the class DIO.

\begin{theorem}\label{thm:dio_pure}
The deterministic pure-state transformation $\psi \to \phi$ is possible under DIO if and only if $\Delta(\psi) \mlt \Delta(\phi)$.
\end{theorem}
We refer to the Appendix for the full proof of the Theorem.
This establishes DIO as another class of operations in which pure-state transformations are fully governed by majorization theory, and reveals an operational equivalence between DIO, IO, and SIO in manipulating pure states. The equivalence is non-trivial: the class DIO is incomparable with IO~\cite{chitambar_2016-1}, and there exist coherence monotones which can increase under DIO despite always decreasing under the action of SIO/IO \cite{bu_2017-1}.

The Theorem immediately lets us apply a plethora of results to coherence manipulation under DIO. For instance, the recent investigation of moderate-deviation interconversion rates under majorization in \cite{korzekwa_2019,chubb_2019} allows one to precisely characterize DIO transformations beyond the single-shot regime; similarly, a recent investigation of quantum coherence fluctuation relations \cite{morris_2018} relies purely on the theory of majorization, and our result immediately establishes that the results can be directly applied to describe the fluctuations and battery-assisted transformations under DIO.

The result can also be extended to so-called heralded probabilistic transformations, where a state $\ket\psi$ is transformed to one of the states $\{\ket{\phi_j}\}$ with a corresponding probability $p_j$, and the final state is identified unambiguously by a classical flag register; formally, the output states take the form $\phi_j \otimes \proj{j}$. One can similarly show the following.
\begin{proposition}\label{prop:dio_pure_prob}
There exists a DIO effecting the transformation $\psi\to\sum_j p_j \phi_j \otimes\proj j$ if and only if
\begin{equation}\begin{aligned}
	\Delta(\psi) \mlt \sum_j p_j \Delta(\phi_j).
\end{aligned}\end{equation}
\end{proposition}
We refer the Reader to the Appendix for details. This again establishes an equivalence between DIO, IO, and SIO in such transformations, extending earlier partial results~\cite{fang_2018}.

\section{Coherence manipulation with $\rho$-DIO}
The existence of a $\rho$-DIO transformation between states $\rho$ and $\sigma$ is equivalent to the existence of a quantum channel $\Lambda$ such that $\Lambda(\rho) = \sigma$ and $\Lambda(\Delta(\rho)) = \Lambda(\Delta(\sigma))$. This has strong connections with the concept of relative majorization \cite{buscemi_2012,buscemi_2017,renes_2016}, and may at first sight suggest that majorization will also play a role in $\rho$-DIO transformations, making them no more powerful than DIO. We will show that this is in fact not the case. To investigate this problem, we now focus on the fundamental tasks of distillation and dilution.

\subsection{Distillation}\label{sec:distillation}
The $\ve$-error one-shot distillable coherence under the class $\rho$-DIO is defined to be the largest size of the maximally coherent state $\ket{\Psi_m} = \sum_i \frac{1}{\sqrt{m}}\ket{i}$ reachable to within an error $\ve$ under a single $\rhodio$ transformation; formally, we have
\begin{equation*}\begin{aligned}
  &C_{d,\rhodio}^{(1),\ve}(\rho) \coloneqq \log \max \lsetr \!m\! \barr \max_{\Lambda \in \rhodio} F(\Lambda(\rho), \Psi_m) \geq 1- \ve\!\rsetr
\end{aligned}\end{equation*}
where $F(\rho,\sigma) = \norm{\sqrt{\rho}\sqrt{\vphantom{\rho}\sigma}}{1}^2$ is the fidelity.
Our first result exactly characterizes this quantity in terms of the hypothesis testing relative entropy $D^\ve_H$, defined as \cite{buscemi_2010,wang_2012,tomamichel_2013}
 \begin{equation*}\begin{aligned}
&D_H^\ve(\rho||\sigma) \!\coloneqq \!-\log\min \{  \Tr M \sigma \;|\; 0\leq M\leq \id, \; 1\!-\! \Tr M \rho \leq\ve \}.
\end{aligned}\end{equation*}
This quantity finds use in the fundamental task of quantum hypothesis testing~\cite{hayashi_2016,hayashi_2017}, where one is interested in distinguishing between two quantum states $\rho$ and $\sigma$ by a measurement $\{M, \id - M\}$, with $D^\ve_H(\rho\|\sigma)$ quantifying exactly the smallest probability of incorrectly accepting the hypothesis of possessing state $\rho$ as true ($\Tr M \sigma$) while constraining the probability of incorrectly accepting the hypothesis of being in possession of state $\sigma$ as true ($\Tr (\id - M) \rho$) to be at most $\ve$. We remark that $D_H^\ve(\rho||\sigma)$ is efficiently computable as a semidefinite program~\cite{dupuis_2012}.

We relate the hypothesis testing relative entropy with distillation in the following.
\begin{theorem}\label{thm:rhodio_dis}
The $\ve$-error one-shot distillable coherence under $\rho$-DIO for any input state $\rho$ is given by
\begin{equation}\begin{aligned}\label{eq:rhodio_dist}
  C_{d,\rhodio}^{(1),\ve}(\rho) = \logfloor{ D^\ve_H (\rho \| \Delta(\rho)) },
\end{aligned}\end{equation}
where $\logfloor{x} \coloneqq \log \lfloor 2^x \rfloor$.
\end{theorem}
This formally establishes a property of the class of operations $\rhodio$ that one might intuitively expect: the more distinguishable a state $\rho$ is from its dephased version $\Delta(\rho)$, the more coherence we can extract from it using $\rho$-DIO. This is indeed very natural in the framework of quantum coherence, as it gives explicit operational meaning to the information contained in the off-diagonal elements of the density matrix $\rho$. It is instructive to compare Eq.~\eqref{eq:rhodio_dist} with the expression for one-shot distillable coherence under DIO~\cite{regula_2017}, where $D^\ve_H$ additionally has to be optimized over a set of operators, and does not enjoy an exact interpretation in this context.

Of particular importance will be the case $\ve=0$, that is, exact deterministic distillation of coherence. The result then reduces to
\begin{equation}\begin{aligned}\label{eq:dist_zero}
  C_{d,\rhodio}^{(1),0}(\rho) &= \logfloor{ D^0_H (\rho \| \Delta(\rho)) } = \log \left\lfloor \frac{1}{\Tr \Pi_\rho \Delta(\rho)} \right\rfloor
\end{aligned}\end{equation}
where $\Pi_\rho$ is the projection onto the support of $\rho$. In particular, combining the results of Thms.~\ref{thm:dio_pure} and \ref{thm:rhodio_dis}, we have the following.

\begin{corollary}\label{corr:oneshot_dist}
A pure state $\ket\psi = \sum_i \psi_i \ket{i}$ can be deterministically transformed to $\ket{\Psi_m}$ under DIO iff 
 \begin{equation}\begin{aligned}
  \max_i |\psi_i|^2 \leq \frac1m,
 \end{aligned}\vspace*{-.3\baselineskip}\end{equation}
while the transformation is possible under $\psi$-DIO iff
\begin{equation}\begin{aligned}
 \braket{ \psi | \Delta(\psi) | \psi} = \sum_i |\psi_i|^4 \leq \frac1m.
\end{aligned}\end{equation}
\end{corollary}
A detailed derivation can be found in the Appendix. The above allows us to easily construct examples of states such that, even though $\ket\psi \to \ket{\Psi_m}$ is impossible under DIO, the transformation can be achieved by $\psidio$. Consider for example the state $\ket\psi \coloneqq \left( \sqrt{\frac58}, \sqrt\frac{3}{16}, \sqrt\frac{3}{16} \right)^T$, for which it can be verified that $\Delta(\Psi_2) \nsucc \Delta(\psi)$, which means the transformation $\ket\psi \to \ket{\Psi_2}$ is impossible by DIO (and in fact by all MIO \cite{regula_2017}). However, we easily compute $\sum_i |\psi_i|^4 = \frac{59}{128} < \frac{1}{2}$, so $C_{d,\psidio}^{(1),0}(\psi) = 1$ and hence one coherence bit $\Psi_2$ can be distilled exactly. This explicitly shows an operational advantage provided by the operations $\rho$-DIO over DIO in state transformations and in particular in coherence distillation. Such an advantage is rather surprising: to any classical observer, the distillation protocol applied to the state $\rho$ is indistinguishable from a classical operation, yet it can distill more coherence than DIO or even the powerful class MIO. Furthermore, due to the majorization condition of Thm.~\ref{thm:dio_pure}, one can also show that $\rho$-DIO has strictly larger capabilities in distillation than even DIO assisted by a pure catalyst state. Furthermore, due to the majorization condition of Thm.~\ref{thm:dio_pure}, it is easy to see that the exact distillation of coherence under DIO cannot be enhanced by the use of a catalyst --- that is, $\psi \otimes \phi \to \Psi_m \otimes \phi$ is possible under DIO if and only if $\psi \to \Psi_m$ is possible --- which then shows that $\rho$-DIO has strictly larger capabilities than even catalysis-assisted DIO.

However, consider now the many-copy scenario in which we can perform joint quantum operations on the composite system $\rho^{\otimes n}$. In the asymptotic limit of independent and identically distributed (i.i.d.) systems, one can then define the distillable coherence under $\rho$-DIO as
\begin{equation}\begin{aligned}
  &C^\infty_{d,\rhodio} (\rho) = \lim_{\ve \to 0} \lim_{n \to \infty} \frac1n C_{d,\rho^{\otimes n}\text{\rm-DIO}}^{(1),\ve}(\rho^{\otimes n}).
\end{aligned}\end{equation}
A simple application of Thm. 1 together with the quantum Stein's lemma \cite{ogawa_2000,ogawa_2004} reveals that we have in fact $C^\infty_{d,\rhodio} (\rho) = D(\rho \| \Delta(\rho))$,
that is, the relative entropy of coherence $D(\rho \| \Delta(\rho))$ characterizes the asymptotic rate of coherence distillation under $\rho$-DIO. But it is known already that under DIO we also have $C^\infty_{d,DIO} (\rho) =  D(\rho \| \Delta(\rho))$ \cite{regula_2017,chitambar_2018},
which means that $\rho$-DIO do not perform any better than DIO in the asymptotic limit. Taking into consideration the operational gap between the operations DIO and $\rho$-DIO in single-shot transformations, the asymptotic result can be quite surprising, since it effectively shows that the advantage provided by $\rho$-DIO over DIO will be relatively minor and will disappear completely at the asymptotic level.

Finally, one can define the zero-error distillable coherence as
\begin{equation}\begin{aligned}
  C^{\infty,0}_{d,\rhodio} (\rho) = \lim_{n \to \infty} \frac1n C_{d,\rho^{\otimes n}\text{\rm-DIO}}^{(1),0}(\rho^{\otimes n}).
\end{aligned}\end{equation}
Noting the additivity of $D^0_H(\rho\|\Delta(\rho))$, from Eq.~\eqref{eq:dist_zero} we immediately get that $C_{d,\rhodio}^{\infty,0}(\rho) = - \log \Tr \Pi_\rho \Delta(\rho)$.

\subsection{Dilution}
Consider the transformation of a maximally coherent state $\Psi_m$ into a general state $\rho$, using a $\Psi_m$-DIO protocol. The one-shot coherence cost is given by
\begin{equation*}\begin{aligned}
  &C_{c,\psimdio}^{(1),\ve}(\rho) \!\coloneqq \log \min \lsetr \!m\!  \barr \!\max_{\Lambda \in \psimdio} F(\Lambda(\Psi_m), \rho) \geq \!1-\ve\! \rsetr.
\end{aligned}\end{equation*}
To characterize this quantity, we will consider the coherence monotone based on the max-relative entropy between $\rho$ and $\Delta(\rho)$ \cite{chitambar_2016-1}, given by
\begin{equation*}\begin{aligned}
  R_\Delta(\rho) \coloneqq& \min \lset \lambda \bar \rho \leq (1+\lambda) \Delta(\rho) \rset
\end{aligned}\end{equation*}
It is easy to verify that $R_\Delta(\Lambda(\rho)) \leq R_\Delta(\rho)$ for any $\rho$-DIO operation $\Lambda$. Using this quantity, we have the following.

\begin{theorem}\label{thm:rhodio_dil}
The $\ve$-error one-shot coherence cost under $\Psi_m$-DIO is given by 
 \begin{align*}
   C_{c,\psimdio}^{(1),\ve}(\rho) \!=\! \log \left\lceil \min \lset R_\Delta (\omega) + 1 \!\bar\! \omega \in \DD,\,F(\rho, \omega) \!\geq\! 1\!-\! \ve \rset \right\rceil,\nonumber
\end{align*}
where $\DD$ denotes the set of all density matrices.
\end{theorem}
Interestingly, comparing the above with the results obtained previously for DIO \cite{zhao_2018}, we have that
 \begin{equation}\begin{aligned}
   C_{c,DIO}^{(1),\ve}(\rho) &= C_{c,\psimdio}^{(1),\ve}(\rho);
 \end{aligned}\end{equation}
that is, the operations $\Psi_m$-DIO provide no advantage over DIO whatsoever. Combining this with the fact that the asymptotic coherence cost under DIO is given exactly by $D(\rho\|\Delta(\rho))$ \cite{chitambar_2018}, we similarly have that
\begin{equation}\begin{aligned}
   C^\infty_{d,\psimdio} (\rho) &= \lim_{\ve \to 0} \lim_{n \to \infty} \frac1n C_{c,\Psi_m\text{\rm-DIO}}^{(1),\ve}(\rho^{\otimes n}) = D(\rho\|\Delta(\rho)).
\end{aligned}\end{equation}

We can also note that the zero-error coherence cost under $\Psi_m$-DIO (or DIO) is given exactly by $ \log \lceil R_\Delta(\rho)+1 \rceil$, and noticing the multiplicativity of $R_\Delta + 1$ one can see that the asymptotic zero-error cost of coherence dilution will be given simply by $\log (R_\Delta(\rho) + 1)$.

\subsection{General transformations and monotones}
When discussing asymptotic state transformations, one is in particular interested in the largest rate $R(\rho \to \sigma)$ at which copies of $\rho$ can be transformed to copies of $\sigma$ under the given class of operations. Our results can be used to show that the rate of any such transformation under $\rho$-DIO is completely characterized by the relative entropy between the states and their diagonals. 

To see this, notice first that any transformation $\rho \to \omega \to \sigma$ such that the operation taking $\rho$ to $\omega$ is $\rho$-DIO and the operation taking $\omega$ to $\sigma$ is $\omega$-DIO results in an overall protocol $\rho \to \sigma$ which is $\rho$-DIO. This allows us to employ maximally coherent states $\Psi_m$ as an intermediary in coherence transformations. Using the fact that $\Psi_{2^m} = \proj{+}^{\otimes m}$, we can interpret the distillable coherence $C^\infty_{d,\rhodio} (\rho)$ as the rate $R(\rho \to \proj{+})$ under $\rho$-DIO, and the coherence cost $C^\infty_{c,\psimdio} (\sigma)$ as the rate $1/R(\proj{+} \to \sigma)$ under $\proj{+}$-DIO; a straightforward argument in analogy with~\cite{horodecki_2003-2} then shows the following.
\begin{corollary}
For any states $\rho$ and $\sigma$, the maximal rate of the asymptotic transformation $\rho \to \sigma$ under $\rho$-DIO operations is given by
\begin{equation}\begin{aligned}
   R(\rho \to \sigma) = \frac{D(\rho\|\Delta(\rho))}{D(\sigma\|\Delta(\sigma))}.
 \end{aligned}\end{equation}
\end{corollary}
\noindent As this is true also for DIO~\cite{chitambar_2018}, asymptotically, $\rho$-DIO provide no advantage whatsoever over DIO in any state transformation.

Furthermore, one can obtain various useful sufficient conditions for the transformations in the one-shot setting. For instance, we show that the monotone $R_\Delta$ can also be used to characterize state transformations under $\rho$-DIO which go beyond coherence distillation and dilution.
\begin{proposition}\label{prop:suff_Rdelta}
If $R_\Delta(\sigma) + 1 \leq 1/ \Tr \Pi_\rho \Delta(\rho)$, then there exists a $\rho$-DIO map such that $\Lambda(\rho) = \sigma$.
\end{proposition}
In the particularly simple case of single-qubit transformations, we furthermore establish an equivalence of $\rho$-DIO and DIO.
\begin{proposition}\label{prop:qubit_trans}
For any single-qubit states $\rho$ and $\sigma$, the transformation $\rho \to \sigma$ is possible under $\rho$-DIO if and only if it possible under DIO, which holds if and only if~\cite{chitambar_2016-1} $R_\Delta(\rho) \geq R_\Delta(\sigma)$ and $\norm{\rho}{\ell_1} \geq \norm{\sigma}{\ell_1}$.
\end{proposition}
\noindent Here, $\norm{\rho}{\ell_1} = \sum_{i,j} \cbraket{i|\rho|j}$. We note from~\cite{chitambar_2016-1} that single-qubit DIO transformations have been shown to be equivalent to both MIO and SIO transformations, and our result thus extends this equivalence also to $\rho$-DIO. This does not hold beyond dimension 2, as we have demonstrated in Sec.~\ref{sec:distillation} a transformation from a qutrit to a qubit system achievable with $\rho$-DIO but impossible under MIO and DIO.

Necessary conditions for single-shot $\rho$-DIO transformations can be characterized by monotones under this class, i.e., functions which obey the property that if there exists a transformation $\rho \to \sigma$ under $\rho$-DIO, then $f(\rho) \geq f(\sigma)$. Some DIO monotones discussed e.g. in~\cite{chitambar_2016-1} will in fact also be $\rho$-DIO monotones --- this includes $R_\Delta$ or the relative entropy $D(\rho\|\Delta(\rho))$. Indeed, any divergence which satisfies the data-processing inequality will form a $\rho$-DIO monotone. Importantly, this includes R\'enyi relative entropies $D_\alpha(\rho\|\Delta(\rho)) = \frac{1}{\alpha-1} \log \Tr \rho^\alpha \Delta(\rho)^{1-\alpha}$, but only in the range $\alpha \in [0,2]$~\cite{petz_1986}, which contrasts with the set DIO for which all $\alpha$ give a valid monotone~\cite{chitambar_2016-1}. For a pure state, the R\'enyi relative entropies reduce to $D_\alpha(\psi\|\Delta(\psi)) = S_{2-\alpha} (\psi)$ \cite{chitambar_2016-1}  where $S_{\gamma}(\psi) = \frac{\gamma}{1-\gamma} \log \norm{\Delta(\psi)}{\ell_\gamma}$ are the R\'enyi entropies. This shows in particular that $\ell_p$ norms of the squared moduli of the coefficients of a pure state are $\psi$-DIO monotones for $p$ in the range $p \in [0,1]$ and reverse monotones for $p \in [1,2]$. An outstanding question is whether such monotones form a complete set, in the sense that the inequality $f(\rho) \geq f(\sigma)$ for each monotone $f$ implies that there exists a $\rho$-DIO transformation taking $\rho$ to $\sigma$. Although a complete set of infinitely many monotones can be defined~\cite{buscemi_2015,buscemi_2016,takagi_2019}, it is in unclear if there exists a finite set of conditions fully characterizing transformations under $\rho$-DIO.

\section{Discussion}
In this work, we tackled the fundamental question of how to efficiently manipulate the resource of quantum coherence under operations which do not use coherence and thus, to a classical observer, appear classical. We studied this question under two classes of channels: DIO, respecting dephasing covariance for any input state, and $\rho$-DIO, tailored to a specific input state. We first shed light on the operational power of DIO and explicitly characterized pure-state transformations under this class, revealing a novel relation between DIO and majorization theory and thereby connecting DIO to several other classes of free operations.
To push the characterization of coherence manipulation under coherence non-detecting operations to its very limit, we introduced the class of operations $\rho$-DIO and investigated the advantages that this extension provides. We showed in particular that, even though the coherence of the input state is not detected, $\rho$-DIO allow one to distill more coherence than DIO in the one-shot setting. Despite $\rho$-DIO constituting a significant relaxation of the constraints of DIO, the increased capabilities of such channels are limited to non-asymptotic regimes --- we showed the advantages to disappear completely at the asymptotic level, where both sets of operations achieve the same performance in all transformations. This suggests that the simpler class $\rho$-DIO closely approximates the performance of all DIO and no significant operational advantage can be obtained by tailoring the coherence non-detecting restriction to a specific input.
The results provide insight into the structure of the ultimate physical constraints on coherence manipulation with free operations and establish new connections in the operational description of quantum coherence.


\begin{acknowledgments}

We acknowledge discussions with Ludovico Lami. This work was supported by the National Research Foundation of Singapore Fellowship No. NRF-NRFF2016-02, the National Research Foundation and L'Agence Nationale de la Recherche joint Project No. NRF2017-NRFANR004 VanQuTe, the program for FRIAS-Nagoya IAR Joint Project Group, and the Japan Society for the Promotion of Science (JSPS) KAKENHI Grant No.19H04066.

\emph{Note}. --- During the completion of this work, Wang and Wilde~\cite{wang_2019} as well as Buscemi et al.~\cite{buscemi_2019} studied the distinguishability of pairs of states in an operational setting and independently obtained results which overlap with parts of this manuscript. Note, however, that although our characterization of $\rho$-DIO relies precisely on the distinguishability between $\rho$ and $\Delta(\rho)$, the tasks of coherence distillation and dilution studied herein are different from the operational framework of~\cite{wang_2019} concerned with manipulating distinguishability.
\end{acknowledgments}

\bibliographystyle{apsrev4-1}
\bibliography{main}

\appendix

\section{Pure-state transformations under DIO}

Let $\left\{\ket x\right\}_{x=1\dots d_\mathrm{in}}\subset\cH_\mathrm{in}$ and $\left\{\ket y\right\}_{y=1\dots d_\mathrm{out}}\subset\cH_\mathrm{out}$ be the incoherent bases on the input and output Hilbert spaces respectively.
\begin{definition}[DIO]
A channel $\cE:\cL\left(\cH_\mathrm{in}\right)\to\cL\left(\cH_\mathrm{out}\right)$ is a \emph{dephasing\hyp covariant incoherent operation (DIO)} if $\Delta\circ\cE=\cE\circ\Delta$, where $\Delta(\cdot)$ is the dephasing channel with respect to the incoherent basis on the corresponding system.
\end{definition}
We now cast the DIO property of a channel in terms of its Kraus operator representations. Let $\cE(\cdot)=\sum_{i=1}^nK_i(\cdot)K_i^\dagger$; we do not have to worry about the value of $n$. The equivalence of $\Delta\circ\cE$ and $\cE\circ\Delta$ can be translated to the equality of their Choi operators:
\begin{align}
&\mathrm{id}\otimes\left[\Delta\circ\cE\right]\left(\sum_{x_1,x_2}\ket{x_1x_1}\bra{x_2x_2}\right) \nonumber\\
= \,&\mathrm{id}\otimes\left[\cE\circ\Delta\right]\left(\sum_{x_1,x_2}\ket{x_1x_1}\bra{x_2x_2}\right)\\
\Rightarrow \quad & \sum_{i,x_1,x_2,y}\ket{x_1}\bra{x_2}\otimes\left(\proj yK_i\ket{x_1}\bra{x_2}K_i^\dagger\proj y\right)\nonumber\\
=\,& \sum_{i,x}\proj x\otimes\left(K_i\proj xK_i^\dagger\right)\nonumber.
\end{align}
This leads to the following handy properties of any Kraus operator representation of a DIO:
\begin{obsm}
Define the vectors $\vect K(y,x)\in\cV\equiv\bbC^n$, $d_\mathrm{in}d_\mathrm{out}$ in number, as follows:
\be
\vect K(y,x):=\left(\bra yK_1\ket x,\bra yK_2\ket x\dots,\bra yK_n\ket x\right).
\ee
Also define $S_{y|x}:=\left\langle\vect K(y,x),\vect K(y,x)\right\rangle$, where $\langle\cdot,\cdot\rangle$ is the standard Hermitian inner product on $\bbC^n$. Then, the following conditions together capture the DIO property of the CP map $\cE(\cdot)=\sum_{i=1}^nK_i(\cdot)K_i^\dagger$:
\begin{enumerate}
\item Diagonal (i.e., incoherent) input produces diagonal output on average: $\left\langle\vect K(y,x),\vect K(y_1,x)\right\rangle=S_{y|x}\delta_{yy_1}$.\label{dyy}
\item Diagonal\hyp free input produces diagonal\hyp free output on average: $\left\langle\vect K(y,x),\vect K(y,x_1)\right\rangle=S_{y|x}\delta_{xx_1}$.\label{dxx}
\item Trace is preserved: $\sum_{y=1}^{d_\mathrm{out}}S_{y|x}=1$ for all $x\in\left\{1,2\dots,d_\mathrm{in}\right\}$; in other words, the matrix $S$ with components $S_{yx}:=S_{y|x}$ is column\hyp stochastic (justifying the ``conditional'' notation).\label{TP}
\end{enumerate}
It is important to bear in mind that each of these conditions involves summing over all the $n$ Kraus operators. In particular, if the ``on average'' condition \ref{dyy} were tightened to apply to each Kraus operator separately while removing \ref{dxx} altogether, the resulting conditions would characterize the class of \emph{incoherent operations (IO)}. Likewise, if both \ref{dyy} and \ref{dxx} were tightened to apply to each Kraus operator, we would have \emph{strictly incoherent operations (SIO)}.
\end{obsm}
For convenience, we also define the normalized vectors $\hat{\vecg\kappa}(y,x):=\fr1{\sqrt{S_{y|x}}}\vect K(y,x)$, whence $\left\langle\hat{\vecg\kappa}(y,x),\hat{\vecg\kappa}(y_1,x)\right\rangle=\delta_{yy_1}$ and $\left\langle\hat{\vecg\kappa}(y,x),\hat{\vecg\kappa}(y,x_1)\right\rangle=\delta_{xx_1}$. In cases where $S_{y|x}=0$, just define $\hat{\vecg\kappa}(y,x)$ to be some unit vector orthogonal to the rest, suitably expanding the space.

\subsection{Deterministic pure\hyp to\hyp pure state conversion}
Now consider the problem of determining the conditions under which there is a DIO deterministically mapping a given pure state $\psi$ to another, $\phi$. We shall prove the following theorem. Say the DIO given by $\vect K(y,x)$ achieves the desired transformation. Then,
\begingroup
\renewcommand\thetheorem{\ref{thm:dio_pure}}
\begin{theorem}\label{thdDIO}
A given initial state $\ket\psi=\sum_x\mu_x\ket x$ can be mapped deterministically to a given target state $\ket\phi=\nu_y\ket y$ by a DIO if and only if the majorization relation
\be\label{eq:majd}
\vect p\prec\vect q
\ee
holds, where $p_x:=\left|\mu_x\right|^2$ and $q_y:=\left|\nu_y\right|^2$.
\end{theorem}
\endgroup
\begin{remark}[Note on majorization]For a vector $\vect v$ on a finite\hyp dimensional real vector space, define $\vect v^{\downarrow}$ as the vector whose components are the components of $\vect v$ arranged in non\hyp increasing order. For example, $(2,-5,2,4)^{\downarrow}=(4,2,2,-5)$. For a pair of vectors $(\vect u,\vect v)$, the majorization relation $\vect u\prec\vect v$ (``$\vect v$ majorizes $\vect u$'') is then defined as the conjunction of the following conditions:
\begin{align}
    u^\downarrow_1&\le v^\downarrow_1;\nonumber\\
    u^\downarrow_1+u^\downarrow_2&\le v^\downarrow_1+v^\downarrow_2;\nonumber\\
    &\vdots\nonumber\\
    u^\downarrow_1+u^\downarrow_2\dots+u^\downarrow_d&\le v^\downarrow_1+v^\downarrow_2\dots+v^\downarrow_d,
\end{align}
where $d$ is the larger of the dimensionalities of $u$ and $v$ (we append the shorter vector with a suitable number of trailing zeros).

A real matrix with nonnegative entries is column\hyp stochastic (row\hyp stochastic) if each of its columns (rows) adds up to 1; a matrix that is both row\hyp and column\hyp stochastic is said to be bistochastic. A nonnegative real matrix is sub\hyp *stochastic if another nonnegative matrix can be added to it to make it *stochastic (``*'' can stand for ``column\hyp'', ``row\hyp'', ``bi\hyp'', or the absence of any of these qualifiers).

For the purposes of proving the above theorem, we will use the following properties of majorization (see e.g.~\cite{bhatia_1996}):
\begin{enumerate}
    \item For normalized probability vectors $\vect u$ and $\vect v$, the relation $\vect u\prec\vect v$ holds if and only if there exists a bistochastic matrix $T$ such that $\vect u=T\vect v$.
    \item For nonnegative vectors $\vect u$ and $\vect v$, if there exists a sub\hyp bistochastic matrix $T$ such that $\vect u=T\vect v$, then $\vect v$ is said to weakly majorize $\vect u$; weak majorization between normalized probability vectors implies non\hyp weak majorization (i.e.\ the previous condition).
\end{enumerate}
\end{remark}
\begin{proof}
Since the overall output is the pure state $\phi$, the output of each individual Kraus operator must necessarily be proportional to $\phi$. In other words, there exists a normalized $\hat{\vect c}\in\bbC^n$ such that
\be
\sum_{x_1=1}^{d_\mathrm{in}}K_i(y,x_1)\mu_{x_1}=c_i\nu_y
\ee
$\forall$ $i,y$. Multiplying both sides by $\overline{K_i(y,x)}$ and summing over $i$,
\begin{align}
\sum_{x_1=1}^{d_\mathrm{in}}\mu_{x_1}\sum_{i=1}^n\overline{K_i(y,x)}K_i(y,x_1)&=\sum_{i=1}^n\overline{K_i(y,x)}c_i\nu_y\nonumber\\
\Rightarrow\sum_{x_1=1}^{d_\mathrm{in}}\mu_{x_1}S_{y|x}\delta_{xx_1}&=\left\langle\vect K(y,x),\hat{\vect c}\right\rangle\nu_y\nonumber\\
\Rightarrow\mu_x\sqrt{S_{y|x}}&=\left\langle\hat{\vecg\kappa}(y,x),\hat{\vect c}\right\rangle\nu_y\\
\Rightarrow p_x\equiv\left|\mu_x\right|^2=\left|\mu_x\right|^2\sum_{y=1}^{d_\mathrm{out}}S_{y|x}&=\sum_{y=1}^{d_\mathrm{out}}\left|\left\langle\hat{\vecg\kappa}(y,x),\hat{\vect c}\right\rangle\right|^2\left|\nu_y\right|^2\nonumber\\&\equiv\sum_{y=1}^{d_\mathrm{out}}\left|\left\langle\hat{\vecg\kappa}(y,x),\hat{\vect c}\right\rangle\right|^2q_y.\nonumber
\end{align}
The second line above follows from condition \ref{dxx}; the following line by dividing throughout by $\sqrt{S_{y|x}}$ and applying the definition of $\hat{\vecg\kappa}(y,x)$; in the last line we just sum the previous line's expressions over $y$ and use the stochasticity of $S$.

Now define $T_{x|y}:=\left|\left\langle\hat{\vecg\kappa}(y,x),\hat{\vect c}\right\rangle\right|^2$. The normalization of $\hat{\vect c}$ and orthonormality of $\left\{\hat{\vecg\kappa}(y,x)\right\}_{x=1}^{d_\mathrm{in}}$ (for each $y$) and $\left\{\hat{\vecg\kappa}(y,x)\right\}_{y=1}^{d_\mathrm{out}}$ (for each $x$) together imply that $T$ is sub\hyp bistochastic. This implies that $\vect p$ is weakly majorized by $\vect q$; normalization of the distributions implies (non\hyp weak) majorization. Incidentally, the same normalization arguments also imply that $n=d_\mathrm{in}$ suffices.

Thus, majorization of the input coherence distribution by the output is necessary for the existence of a DIO deterministically mapping $\psi\mapsto\phi$. Since this condition is already known to be sufficient for the existence of such an SIO, its sufficiency for the existence of such a DIO follows.
\end{proof}

\subsection{Probabilistic pure\hyp to\hyp pure state conversion}
A possible definition of probabilistic conversion of a given state $\ket\psi$ to an ensemble $\left\{\left(\eta_j,\ket{\phi_j}\right)\right\}$ under a class of operations is one where a channel $\Lambda\equiv\left\{K_j\right\}$ belonging to the class satisfies $K_j\psi K_j^\dagger=\eta_j\phi_j$ (possibly with some of the $\phi_j$ mutual duplicates). This definition is not easily amenable to the treatment of the previous section, owing to the lack of individual Kraus operator\n based constraints in DIO. This is in contrast with IO and SIO, whose definitions constrain each Kraus operator. Extension of our results to the pure\hyp state\hyp to\hyp mixed\hyp state case is hindered by this obstacle.

Nevertheless, we can say something useful about \emph{heralded} probabilistic conversion from a pure state $\ket\psi$ to an ensemble $\left\{\left(\eta_j,\ket{\phi_j}\right)\right\}$. This entails that a conversion to $\phi_j$ be heralded by a correlated ``flag'' system whose state unambiguously identifies $j$. In other words, we require a DIO to map $\psi$ to $\sum_j\eta_j\sigma_j\otimes\phi_j$, with the ``flag states'' $\sigma_j$ on the first subsystem unambiguously distinguishable. We might as well set these to some mutually\hyp orthogonal $\proj j$ without loss of generality. To keep the game fair, we shall require $\ket j$ to constitute the axiomatic incoherent basis for the flag system.
\begingroup
\renewcommand\thetheorem{\ref{prop:dio_pure_prob}}
\begin{proposition}
There exists a DIO effecting the transformation $\psi\mapsto\sum_j\eta_j\proj j\otimes\phi_j$ if and only if
\be
\vect p\prec\sum_j\eta_j\vect q_j^{\downarrow},
\ee
where $\vect p$ is as before and $q^j_y:=\left|\nu^j_y\right|^2$ for $\ket{\phi_j}=\sum_y\nu^j_y\ket y$.
\end{proposition}
\begin{proof}
Without loss of generality, we can decompose the requisite DIO using Kraus operators of the form
\be\label{Kflag}
K_{j,m}=\sum_{x,y}K_{j,m}(j,y;x)\ket j\otimes\ket y\bra x,
\ee
with
\be
\sum_mK_{j,m}\psi K_{j,m}^\dagger=\eta_j\proj j\otimes\phi_j.
\ee
Adapting the notation of the previous section, the DIO conditions can be cast as
\begin{enumerate}
\item $\left\langle\vect K(j,y;x),\vect K(j_1,y_1;x)\right\rangle=S_{j,y|x}\delta_{jj_1}\delta_{yy_1}$.
\item $\left\langle\vect K(j,y;x),\vect K(j,y;x_1)\right\rangle=S_{j,y|x}\delta_{xx_1}$.\label{djxx}
\item $\sum_{j,y}S_{j,y|x}=1$ for all $x$.
\end{enumerate}
In fact, considering the form \eqref{Kflag} allows us to strengthen condition \ref{djxx} above to
\be\label{djmxx}
\left\langle\vect K(j,y;x),\vect K(j,y;x_1)\right\rangle_j=S_{j,y|x}\delta_{xx_1},
\ee
where $\left\langle\vect u,\vect v\right\rangle_j:=\sum_m\overline{u_{j,m}}v_{j,m}$ for $\vect u$, $\vect v$ in the abstract vector space $\cV$ defined above. The rest of our proof to Theorem~\ref{thdDIO} can be applied as such, but with all such vectors restricted to the subspace corresponding to a specific $j$. This leads to
\be
p_x\sum_yS_{j,y|x}=\sum_y\left|\left\langle\hat{\vecg\kappa}(j,y;x),\vect c^j\right\rangle\right|^2q^j_y,
\ee
where
\be
K_{j,m}\ket\psi=c^j_m\ket j\otimes\ket{\phi_j}
\ee
with $\eta_j=\sum_m\left|c^j_m\right|^2$. Thus,
\be\label{eq:majc}
\vect p_j\equiv\left(p_x\sum_yS_{j,y|x}\right)_x\precw\eta_j\vect q^j.
\ee
From the properties of the majorization relation, we have
\be
\vect p_j^{\downarrow}\precw\eta_j\vect q_j^{\downarrow}.
\ee
Summing over $j$ yields
\be
\sum_j\vect p_j^{\downarrow}\precw\sum_j\eta_j\vect q_j^{\downarrow},
\ee
which strengthens under normalization considerations to non\hyp weak majorization. But $\vect p=\sum_j\vect p_j\prec\sum_j\vect p_j^{\downarrow}$. Therefore,
\be
\vect p\prec\sum_j\eta_j\vect q_j^{\downarrow}.
\ee
Again, the converse follows from the corresponding result about SIO.
\end{proof}

\section{$\rho$-DIO transformations}

\subsection{Distillation}

Consider the rate of distillation of coherence under the operations $\rho$-DIO, i.e. channels $\Lambda$ such that $\Delta \circ \Lambda(\rho) = \Lambda \circ \Delta(\rho)$ for some fixed $\rho$. We will use $\< A, B \> = \Tr (A^\dagger B)$ for the Hilbert-Schmidt inner product.

\begingroup
\renewcommand\thetheorem{\ref{thm:rhodio_dis}}
\begin{theorem}
The one-shot distillable coherence under $\rho$-DIO for any input state $\rho$ is given by
\begin{equation*}\begin{aligned}
  C_{d,\rhodio}^{(1),\ve}(\rho) = \logfloor{ D^\ve_H (\rho \| \Delta(\rho)) },
\end{aligned}\end{equation*}
where $\logfloor{x} \coloneqq \log \lfloor 2^x \rfloor$.
\end{theorem}
\endgroup
\begin{proof}
 The $\ve$-error one-shot rate of distillation can be expressed as
\begin{equation}\begin{aligned}
  C_{d,\rhodio}^{(1),\ve}(\rho) := \log \max \lset m \in \mathbb N \bar F_\rhodio(\rho,m) \geq 1- \ve \rset
\end{aligned}\end{equation}
where $F_\rhodio$ is the achievable fidelity of distillation, i.e.
\begin{equation}\begin{aligned}
  F_\rhodio(\rho, m) \coloneqq \max_{\Lambda \in \rhodio} F(\Lambda(\rho), \Psi_m)
\end{aligned}\end{equation}
with $\Psi_m$ the $m$-dimensional maximally coherent state.

Defining the twirling $\T(\cdot) = \frac{1}{d!} \sum_{i=1}^{d!} U_{\pi(i)} \cdot U_{\pi(i)}^\dagger$ where each $\pi(i)$ is a permutation of the basis vectors, we have that
\begin{equation}\begin{aligned}
  F_{\rhodio}(\rho,m) &= \max_{\Lambda \in \rhodio} F (\T \circ \Lambda(\rho), \T(\Psi_m)) \\
  &= \max_{\Lambda \in \rhodio} F (\T \circ \Lambda(\rho), \Psi_m),
\end{aligned}\end{equation}
that is, it suffices to optimise over twirled maps of the form $\Lambda = \T \circ \Lambda$. Due to permutation invariance, the output of such a map must satisfy $\braket{i | \Lambda(\cdot) | i} = \braket{\pi(i) | \Lambda(\cdot) | \pi(i)} \; \forall i $ and $\braket{i | \Lambda(\cdot) | j} = \braket{\pi(i) | \Lambda(\cdot) | \pi(j)} \; \forall i\neq j$, where $\pi$ is an arbitrary permutation; imposing these constraints, we can write any such map as
\begin{equation}\begin{aligned}
  \Lambda(Q) &= \< X, Q \> \Psi_m + \< \id - X, Q \> \frac{\id-\Psi_m}{m-1}
\end{aligned}\end{equation}
for some operator $X$. The complete positivity of $\Lambda$ is equivalent to the condition $0 \leq X \leq \id$, and to further impose that $\Lambda \in \rho$-DIO, we need that
\begin{equation}\begin{aligned}
 \frac{\id}{m}  &= \Delta \circ \Lambda(\rho)\\
  &= \Lambda \circ \Delta(\rho)\\
  &= \< X, \Delta(\rho) \> \Psi_m + (1 - \< X, \Delta(\rho)) \> \frac{\id-\Psi_m}{m-1}
\end{aligned}\end{equation}
which is satisfied if and only if $\< X, \Delta(\rho) \> = \frac{1}{m}$. Altogether, we have
\begin{equation}\begin{aligned}\label{eq:fidelity}
  F_\rhodio(\rho,m) = \max \lsetr \< X, \rho \> \barr 0 \leq X \leq \id,\; \< X, \Delta(\rho) \> = \frac{1}{m} \rsetr.
\end{aligned}\end{equation}
Although originally defined for $m \in \NN$, we extend the definition of $F_\rhodio$ to any $m \geq 1$ as in Eq.~\eqref{eq:fidelity}. The one-shot rate of distillation is then
\begin{equation}\begin{aligned}
  C_{d,\rhodio}^{(1),\ve}(\rho) \\
  = &\logfloor{ \log \max \lset m \in \mathbb R \bar F_\rhodio(\rho,m) \geq 1-\ve \rset }\\
  = &\logfloor{ - \log \min \lset \< X, \Delta(\rho) \> \bar \< X, \rho \> \geq 1-\ve,\; 0 \leq X \leq \id \rset }\\
  = &\logfloor{ D^\ve_H (\rho \| \Delta(\rho)) }
\end{aligned}\end{equation}
where $D^\ve_H$ is the hypothesis testing relative entropy.
\end{proof}

By considering the asymptotic scenario, we then have
\begin{equation}\begin{aligned}\label{eq:asymp_distillation}
  C^\infty_{d,\rhodio} (\rho) &= \lim_{\ve \to 0} \lim_{n \to \infty} \frac1n C_{d,\rho^{\otimes n}\text{\rm-DIO}}^{(1),\ve}(\rho^{\otimes n})\\
  &= D(\rho \| \Delta(\rho))
\end{aligned}\end{equation}
by quantum Stein's lemma \cite{ogawa_2000}.

Recall from \cite{regula_2017} that for DIO, we have
\begin{equation}\begin{aligned}
  C_{d,\DIO}^{(1),\ve}(\rho) = \logfloor{ \min_{\substack{X=\Delta(X)\\\Tr X = 1}} D^\ve_H (\rho \| X) }
\end{aligned}\end{equation}
and
\begin{equation}\begin{aligned}
  C^\infty_{d,\DIO} (\rho) = D(\rho \| \Delta(\rho))
\end{aligned}\end{equation}
which shows that, although $\rho$-DIO can have a larger rate of one-shot distillation than DIO, asymptotically the operations have the same capabilities in distillation.

\subsubsection{Exact distillation}

Consider now the case of zero-error distillation, that is, exact transformations $\rho \to \Psi_m$ under $\rhodio$. From the above, we have that the rate of zero-error distillation is given by
\begin{equation}\begin{aligned}
  C_{d,\rhodio}^{(1),0}(\rho) &= \logfloor{ D^0_H (\rho \| \Delta(\rho)) }\\
  &= \log \left\lfloor \frac{1}{\< \Pi_\rho, \Delta(\rho) \>} \right\rfloor
\end{aligned}\end{equation}
where $\Pi_\rho$ is the projection onto the support of $\rho$. 

This simplifies in particular for the case of a pure state $\ket\psi = \sum_i \psi_i \ket{i}$:
\begin{equation}\begin{aligned}
  C_{d,\psidio}^{(1),0}(\psi) &= \log \left\lfloor \braket{\psi | \Delta(\psi) | \psi}^{-1} \right\rfloor\\
  &= \log \left\lfloor \left(\Tr \Delta(\psi)^2 \right)^{-1} \right\rfloor\\
  &= \log \left\lfloor \left( \sum_i |\psi_i|^4 \right)^{-1} \right\rfloor.
\end{aligned}\end{equation}
Explicitly, we have that
\begin{equation}\begin{aligned}
  \psi \to \Psi_m \; \iff \; \sum_i |\psi_i|^4 \leq \frac1m.
\end{aligned}\end{equation}

Notice that the state $\ket\psi \coloneqq \left( \sqrt{\frac58}, \sqrt\frac{3}{16}, \sqrt\frac{3}{16} \right)^T$ from the main text gives an explicit example of a case where coherence monotones under IO and DIO can \textit{increase} in the $\rho$-DIO transformation: specifically, consider the monotone $C_2(\ket\psi) = \sum_{i=2}^{d} |\psi^\downarrow_i|^2$ where $\psi^\downarrow_i$ denote coefficients of $\ket\psi$ arranged in non-increasing order by magnitude \cite{vidal_2000,zhu_2017-1}. For this monotone, with $\ket\psi$ as above, we have $C_2(\ket\psi) = \frac38$ but $C_2(\ket{+}) = \frac12$.

Finally, we remark the additivity of $-\log \Tr \Pi_\rho \Delta(\rho)$ in the sense that
\begin{equation}\begin{aligned}
  - \log \< \Pi_{\rho^{\otimes n}}, \Delta(\rho^{\otimes n}) \> &= - \log \< \Pi_\rho^{\otimes n}, \Delta(\rho)^{\otimes n} \>\\
  & =  - n \log \< \Pi_\rho, \Delta(\rho) \>
\end{aligned}\end{equation}
which in particular gives the asymptotic zero-error distillable coherence as
\begin{equation}\begin{aligned}
  C_{d,\rhodio}^{\infty,0}(\rho) = - \log \< \Pi_\rho, \Delta(\rho) \>.
\end{aligned}\end{equation}

\subsection{Dilution}

\begingroup
\renewcommand\thetheorem{\ref{thm:rhodio_dil}}
\begin{theorem}
The one-shot coherence cost under $\Psi_m$-DIO operations is given by 
\begin{align}
 C_{c,\psimdio}^{(1),\ve}(\rho) \!=\! \log \left\lceil \min \lset R_\Delta (\omega) + 1 \!\bar\! \omega \in \DD,\,F(\rho, \omega) \!\geq\! 1\!-\! \ve \rset \right\rceil.\nonumber
\end{align}
\end{theorem}
\endgroup
\begin{proof}
Using a twirling argument similar to the distillation case, we can without loss of generality limit ourselves to operations of the form $\Lambda = \Lambda \circ \T$, which take the form
\begin{equation}\begin{aligned}
   \Lambda(Q) &= \< \Psi_m, Q \> X + \< \id - \Psi_m, Q \> Z
\end{aligned}\end{equation}
for some operators $X, Z$. The completely positivity and trace preservation of $\Lambda$ impose that $X, Z$ are valid quantum states. To impose that $\Lambda \in \Psi_m$-DIO, we have that
\begin{equation}\begin{aligned}
  \Delta(X) &= \Delta \circ \Lambda (\Psi_m)\\
  &= \Lambda \circ \Delta (\Psi_m)\\
  &= \frac1m X + \frac{m-1}{m} Z.
\end{aligned}\end{equation}
Noticing that $\Lambda(\Psi_m) = X$, this means that the set of states $\omega$ such that $\omega = \Lambda(\Psi_m)$ for some $\Psi_m$-DIO protocol $\Lambda$ is precisely the set of states for which there exists a state $\sigma$ such that $\frac1m \omega + \frac{m-1}{m} \sigma = \Delta(\omega)$. It is easy to see that this is only possible if $\Delta(\sigma) = \Delta(\omega)$. Defining the function
\begin{equation}\begin{aligned}
  g(\omega) \coloneqq \min \lsetr \lambda \barr \frac{\omega + \lambda \sigma}{1+\lambda} \in \I,\; \sigma \in \DD,\; \Delta(\sigma) = \Delta(\omega) \rsetr
\end{aligned}\end{equation}
it is not difficult to show that, in fact, $g(\omega) = R_\Delta(\omega)$~\cite{chitambar_2016-1}. We then have that the one-shot coherence cost under $\Psi_m$-DIO is
\begin{equation}\begin{aligned}
  &C_{c,\psimdio}^{(1),\ve}(\rho) \\
  = &\log \left\lceil \min \lset m \in \mathbb R \bar \omega \in \DD,\; F(\rho, \omega) \geq 1\!-\! \ve,\; R_\Delta (\omega) \leq \!m\!-\!1 \rset \right\rceil\\
  = &\log \left\lceil \min \lset R_\Delta (\omega) + 1 \bar \omega \in \DD,\;F(\rho, \omega) \geq 1- \ve \rset \right\rceil
\end{aligned}\end{equation}
as required.
\end{proof}

\subsubsection{Exact dilution}

Zero-error dilution can be characterised straightforwardly as the required coherence cost is simply
\begin{equation}\begin{aligned}
  C_{c,\psimdio}^{(1),0}(\rho) = \logceil{ \log ( R_\Delta(\rho) + 1) } .
\end{aligned}\end{equation}
Note that
\begin{equation}\begin{aligned}
  R_\Delta(\rho) &= \norm{\Delta(\rho)^{-1/2} \rho \Delta(\rho)^{-1/2} }{\infty}-1
\end{aligned}\end{equation}
which makes this quantity easy to express and compute. In particular, noting that for every state we have $\rho \leq \rank(\Delta(\rho)) \Delta(\rho)$~\cite{hayashi_2001}, this straightforwardly shows that the maximally coherent state $\Psi_m$ can be transformed into any other state with $\rank\Delta(\rho) \leq m$, and it acts as a ``golden unit'' under the operations $\Psi_m$-DIO. For a pure state, we have exactly \cite{chitambar_2016-1} 
\begin{equation}\begin{aligned}
  R_\Delta(\psi) = \rank \Delta(\psi) - 1
\end{aligned}\end{equation}
which further simplifies the characterisation.

We remark the additivity of $\log ( R_\Delta(\rho) + 1)$, immediately establishing that the asymptotic zero-error coherence cost under $\Psi_m$-DIO is given by
\begin{equation}\begin{aligned}
  C_{c,\psimdio}^{\infty,0}(\rho) = \log ( R_\Delta(\rho) + 1).
\end{aligned}\end{equation}

\subsection{General transformations}

We can generalise the approach of coherence dilution to give a simple sufficient condition for transformations to general states. For any $\omega \in \DD$, we have the following.
\begingroup
\renewcommand\thetheorem{\ref{prop:suff_Rdelta}}
\begin{proposition}
If $R_\Delta(\omega) + 1 \leq \frac{1}{\<\Pi_\rho, \Delta(\rho) \>}$, then there exists a $\rho$-DIO map such that $\Lambda(\rho) = \omega$.
\end{proposition}
\endgroup
\begin{proof}
Recalling that $R_\Delta(\omega) = \min \lset \lambda \bar \omega \leq (\lambda+1) \Delta(\omega) \rset$, it is easy to see that for any $\lambda \geq R_\Delta(\omega)+1$, there exists a state $\sigma$ such that $\omega + (\lambda-1) \sigma = \lambda \Delta(\omega)$. By assumption, there exists in particular $\sigma \in \DD$ which satisfies
\begin{equation}\begin{aligned}
   \omega + \left(\frac{1}{\< \Pi_\rho, \Delta(\rho) \>}-1\right) \sigma = \frac{1}{\< \Pi_\rho, \Delta(\rho) \>} \Delta(\omega).
 \end{aligned}\end{equation} 
With this choice of $\sigma$, define the map
\begin{equation}\begin{aligned}
  \Lambda(X) = \< \Pi_\rho , X\> \omega + \< \id - \Pi_\rho, X \> \sigma.
\end{aligned}\end{equation}
This map is clearly CPTP, and we have
\begin{equation}\begin{aligned}
  \Lambda(\rho) &= \omega,\\
  \Lambda(\Delta(\rho)) &= \< \Pi_\rho, \Delta(\rho) \> \omega + ( 1 - \< \Pi_\rho, \Delta(\rho) \> ) \sigma\\
  &= \< \Pi_\rho, \Delta(\rho) \> \left( \omega + \left[ \frac{1}{\< \Pi_\rho, \Delta(\rho) \>} - 1 \right] \sigma \right)\\
  &= \Delta(\omega)
\end{aligned}\end{equation}
as required.
\end{proof}

The above condition is in general not necessary for $\rho$-DIO transformations, and indeed in general $R_\Delta(\rho) + 1 > \frac{1}{\<\Pi_\rho, \Delta(\rho) \>}$ which means that even the trivial transformation $\rho \to \rho$ might not satisfy the condition of the Proposition.

When the input state is a qubit, however, the transformations can be characterized exactly (see also \cite{alberti_1980,heinosaari_2012}). We will in particular establish an equivalence between $\rho$-DIO and DIO in such transformations.
\begingroup
\renewcommand\thetheorem{\ref{prop:qubit_trans}}
\begin{proposition}
For any single-qubit states $\rho$ and $\sigma$, the transformation $\rho \to \sigma$ is possible under $\rho$-DIO if and only if it possible under DIO.
\end{proposition}
\endgroup
\begin{proof}
Clearly, if a DIO transformation $\rho\to\sigma$ exists, then so does a $\rho$-DIO transformation by the inclusion DIO $\subseteq \rho-$DIO. By~\cite[Thm. 30]{chitambar_2016-1}, the DIO transformation is possible if and only if $R_\Delta(\rho) \geq R_\Delta(\sigma)$ and $\norm{\rho}{\ell_1} \geq \norm{\sigma}{\ell_1}$. Since $R_\Delta$ is trivially a $\rho$-DIO monotone as discussed earlier, it suffices to show that $\norm{\cdot}{\ell_1}$ is a $\rho$-DIO monotone, which will mean that the existence of a $\rho$-DIO transformation implies the existence of a DIO transformation.

To see that this is indeed true, note that $\norm{\rho - \Delta(\rho)}{1}$ is clearly a $\rho$-DIO monotone due to the contractivity of the trace distance under CPTP maps. But for a single-qubit state $\rho = \begin{pmatrix}\rho_{00} & \rho_{01}\\\rho_{01}^* & \rho_{11}\end{pmatrix}$, we have $\norm{\rho - \Delta(\rho)}{1} = 2 \left|\rho_{01}\right| = \norm{\rho}{\ell_1} - 1$, which means that $\norm{\rho}{\ell_1}$ is also a $\rho$-DIO monotone for all single-qubit states.
\end{proof}

\end{document}